\documentclass[11pt]{article}
\usepackage[margin=1in]{geometry}
\usepackage{amsmath,amssymb,amsthm}
\usepackage[hidelinks]{hyperref}
\usepackage{microtype}
\usepackage{tikz}

\newtheorem{theorem}{Theorem}[section]
\newtheorem{lemma}[theorem]{Lemma}
\newtheorem{proposition}[theorem]{Proposition}

\newtheorem{definition}[theorem]{Definition}
\theoremstyle{remark}

\newcommand{\red}{\operatorname{red}}
\newcommand{\supp}{\operatorname{supp}}
\newcommand{\spn}{\operatorname{span}}
\newcommand{\cK}{\mathcal K}
\newcommand{\cW}{\mathcal W}
\newcommand{\cC}{\mathcal C}
\newcommand{\cE}{\mathcal E}

\title{Polynomially larger deletion codes\\
by linear hashing of substring counts}
\author{Eyal En Gad}
\date{}

\begin{document}
\maketitle

\begin{abstract}
We reduce the leading redundancy coefficient for binary codes correcting
any fixed number $t\ge2$ of deletions from $2t$ to $2t-1$.
This is the first improvement to that coefficient since Levenshtein's
1965 bound. More generally, over every fixed alphabet of size $q\ge2$,
codes of length $n$ correcting any combination of at most $t$ insertions,
deletions and substitutions exist with redundancy
$(2t-1)\log_q n+O_t(\log_q\log_q n)$.
For two deletions, the upper-bound coefficient drops from $4$ to $3$,
halving the gap to the lower-bound coefficient $2$.
The leading coefficient $2t-1$ also holds for $t$ bounded substring edits,
where both the removed and inserted strings have bounded length.
The proof combines random binary linear hashing of substring counts with a
two-colouring argument.
A linear dependence among differences of substring-count vectors forces
their local changes to overlap, and this overlap lowers the redundancy
coefficient from $2t$ to $2t-1$.
The result provides an edit-distance analogue of Varshamov's 1957 bound
in the Hamming metric, both in the leading redundancy coefficient and
in the role of linear independence in the proof.
\end{abstract}

\section{Introduction}\label{sec:intro}

\subsection{The problem}

Fix a finite alphabet $\Sigma$ of size $q\ge2$. A \emph{word} is a finite
string of letters from $\Sigma$, and its \emph{length} is the number of
letters. A \emph{deletion} removes one letter, an \emph{insertion} adds
one letter, possibly at an end, and a \emph{substitution} replaces one
letter by a different letter. Each is an \emph{edit}. The \emph{edit distance}
between two words is the minimum number of edits needed to transform one
into the other. Fix an integer $t\ge1$, the \emph{error radius}. Two words $y$ and $z$ are
\emph{confusable at radius $t$} if some word can be obtained from each
by at most $t$ edits. Equivalently, their edit distance is at most
$2t$. A set $\cC\subseteq\Sigma^n$
\emph{corrects $t$ edits} if no two distinct words of $\cC$ are confusable
at radius $t$. We call such a set a code; in particular, it corrects $t$
deletions.
The \emph{redundancy} of $\cC$, in $q$-ary symbols, is
\[
   \red(\cC)=n-\log_q|\cC|,
\]
and $\red^*(n)$ is the least redundancy of a set
$\cC\subseteq\Sigma^n$ correcting $t$ edits, with $t$ and $q$ fixed.
We write
$O_t(f(n))$ for a quantity whose absolute value is at most $c_tf(n)$,
with $c_t$ depending only on $t$, for all sufficiently large $n$; the
threshold may depend on both $t$ and $q$. The notation $O_{t,q}$ also
allows the constant to depend on $q$, and subscripts such as $O_{t,b}$
are used in the same way. The question is the growth of
$\red^*(n)$ and, in particular, the coefficient of $\log_qn$.

\subsection{The bounds known}\label{sec:bounds}

For binary deletion codes, Levenshtein~\cite{Levenshtein1966} established
an upper bound on the redundancy with leading coefficient $2t$ in 1965.
A greedy counting argument gives the same leading coefficient $2t$
for codes correcting $t$ edits. Two confusable words can be transformed
into one another by at most $2t$
edits: first reach a common word, then reverse the edits leading
to it from the other word. Each edit is specified by its kind, position
and, for an insertion or substitution, its new letter.
For fixed $t,q$, this gives at most $O_{t,q}(n^{2t})$ words of length $n$
confusable with any given word of length $n$. Form a graph on $\Sigma^n$
by joining distinct confusable words. A code correcting $t$ edits is an
independent set in this graph, meaning a set containing no adjacent pair.
The graph has maximum degree $O_{t,q}(n^{2t})$, so a greedy choice of
vertices gives
\begin{equation}
   \red^*(n)\ \le\ 2t\log_qn+O_{t,q}(1).
\label{eq:lev}
\end{equation}
In the other direction, the deletion sphere-packing bound applies
because every code correcting $t$ edits also corrects $t$ deletions.
For fixed alphabets it gives~\cite[Theorem~7]{SpeeWeber}
\begin{equation}
   \red^*(n)\ \ge\ t\log_qn-O_{t,q}(1).
\label{eq:lower}
\end{equation}
For a single deletion, the binary Varshamov--Tenengolts codes~\cite{VT},
whose deletion-correction property was proved by
Levenshtein~\cite{Levenshtein1966}, and Tenengolts' nonbinary
codes~\cite{Tenengolts1984} attain the optimal coefficient $1$.
Before this work, for every fixed $t\ge2$ and fixed alphabet the best
known leading coefficients for correction of $t$ deletions were $2t$
in the upper bound and $t$ in the lower bound. Section~\ref{sec:related}
discusses the logarithmic improvement in code size of Alon et al.\ and
its refinement by Kong, Tamo and Wei~\cite{KongTamoWei}, together with
explicit constructions
and mixed-error bounds.

\subsection{The result}

\begin{theorem}\label{thm:intro}
For every $t\ge2$ there is a constant $C_t$ such that, for every finite
alphabet $\Sigma$ of size $q\ge2$ and every sufficiently large $n$
(depending on $t,q$), there is a code $\cC\subseteq\Sigma^n$ correcting
any combination of at most $t$ insertions, deletions and substitutions
in total, with
\[
   \red(\cC)\ \le\ (2t-1)\log_qn+C_t\log_q\log_qn.
\]
\end{theorem}

The coefficient of $\log_qn$ drops from $2t$ to $2t-1$, with a correction
term of doubly logarithmic order. In particular, the theorem improves
the deletion-code bound: at $t=2$ it gives redundancy
$3\log_qn+O(\log_q\log_qn)$, compared with the previous upper-bound
coefficient $4$ and the lower-bound coefficient $2$. Section~\ref{sec:main} proves a more general theorem, stated below,
which includes Theorem~\ref{thm:intro} as a special case.

A \emph{$b$-substring edit} replaces a substring $U$ by a word $V$,
where $|U|,|V|\le b$ and $U\ne V$; either word may be empty.
This standard model~\cite{LiGabrysFarnoud} includes bursts of insertions,
deletions and substitutions, as well as mixtures within a bounded window.
We allow at most $t$ substring edits, with no separation
requirement. Two words are \emph{burst-confusable} if some word can be
obtained from each by at most $t$ such edits. A code corrects the errors if
no two distinct codewords are burst-confusable.

\begin{theorem}\label{thm:bursts}
For every pair of integers $t\ge2$ and $b\ge1$, there is a constant $C_{t,b}$ such that,
for every alphabet $\Sigma$ of size $q\ge2$ and all sufficiently large
$n$ depending on $t,b,q$, there is a code $\cC\subseteq\Sigma^n$
correcting at most $t$ $b$-substring edits with
\[
   \red(\cC)\le(2t-1)\log_qn+C_{t,b}\log_q\log_qn.
\]
\end{theorem}

The parameter $b$ bounds both substring lengths, which may differ.
A burst of at most $b$ individual edits within an input substring of
length at most $b$ may increase its length to $2b$ through insertions.
Such a burst is a $2b$-substring edit. Applying the theorem with $2b$
in place of $b$ gives the same leading coefficient $2t-1$ and a
remainder of $O_{t,b}(\log_q\log_qn)$.

We prove Theorem~\ref{thm:bursts} directly. Ordinary insertions,
deletions and substitutions are exactly the $1$-substring edits, so
Theorem~\ref{thm:intro} follows by setting $b=1$. The proof treats
all local changes alike, as the replacement of one short substring by
another. The coefficient depends on the number of bursts,
rather than their lengths. Section~\ref{sec:related} compares this bound
with earlier burst-correcting codes.

The improvement has a classical analogue in the Hamming metric. For
codes over a fixed alphabet of size $q$ correcting a fixed number $t$
of substitutions, Gilbert's greedy bound~\cite{Gilbert1952} gives
redundancy $2t\log_qn+O_t(1)$. When $q$ is a prime power, Varshamov's
linear-code bound~\cite{Varshamov1957} gives
$(2t-1)\log_qn+O_t(1)$. Unlike Varshamov's codes, ours are not linear,
and cannot be: a linear code correcting even one deletion has rate at
most $1/2$~\cite{AbdelGhaffarFerreiraCheng}.

For $t\ge2$, algebraic constructions in the Hamming metric improve on
Varshamov's leading coefficient for every fixed alphabet
size~\cite[Theorem~1]{YekhaninDumer}. Binary BCH
codes~\cite{BoseRayChaudhuri1960,Hocquenghem1959} correcting $t$
substitutions have redundancy $t\log_2n+O_t(1)$, attaining the optimal
leading coefficient $t$.
Whether our coefficient $2t-1$ is optimal for deletion or edit correction
remains open, even in the binary case.

\subsection{Proof overview}\label{sec:idea}

We illustrate the proof with binary codes correcting two deletions,
a case that already captures its main ideas.
Sections~\ref{sec:borrowed} to~\ref{sec:main} prove
Theorem~\ref{thm:bursts} over arbitrary fixed alphabets. Here two words are
confusable if they can produce the
same word after at most two deletions each. Passing through this common
word transforms one into the other by at most two deletions and two
insertions. Choosing the edit positions and inserted bits gives
$O(n^4)$ possible partners per word.

Independent random labels from a range of order
$n^{4}$ make it unlikely that a word shares its label with a confusable
partner, and the most common label among the words with no such partner
gives a code with the coefficient $4$ of~\eqref{eq:lev}. We use about $n$
times fewer labels and allow confusable words to share a label.
After setting aside a few exceptional words,
described below, these pairs form a graph. We discard the words in its
components that contain odd cycles and two-colour the rest; one extra
bit then separates every pair. To keep
the number of discarded words small, we use a structured labelling
computed from the words.

We restrict to words whose sufficiently long substrings occur at most
once. Represent each word by its \emph{spectrum}, the vector counting
its substrings of a chosen sufficiently large length. The restriction on
repeated substrings ensures that this vector is binary and determines
the word. We choose a random binary parity-check matrix and use the
syndrome of each word's spectrum as its label.

For a fixed word in this family, we show that our labelling makes it
unlikely that its component contains an odd cycle. Such a component
contains a closed walk that starts at this word, the \emph{starting
word}, and has an odd number of steps. From this walk we extract a
bounded \emph{witness} of $R\ge2$ spectrum differences, independent
modulo two, each between two words with the same label. An odd cycle
thus forces a witness at the starting word, so a union bound over
witnesses bounds the probability of an odd cycle. Choosing four edit positions independently for each of the $R$
differences would give about $n^{4R}$ possibilities. We show that overlap
reduces the number of witnesses to about $n^{3R}$, apart
from powers of $\log n$ per difference.
This saving in the witness count allows us to lower the leading redundancy
coefficient from $4$ to $3$.

Two deletions and two insertions far apart change the spectrum in four
places, giving four \emph{local changes}. Choosing their positions costs
about $n^4$. Each local change is described by the substrings around the
edit, before and after it. Two changes overlap when their descriptions
contain the same sufficiently long substring. Knowing the substrings
describing one change gives us the letters in the shared part of the other.
By our substring-uniqueness restriction, a known sufficiently long
substring determines its position whenever it occurs in the starting word.
Section~\ref{sec:resolve} uses this to recover each connected group of
overlapping local changes from one freely chosen position in the starting word;
the remaining choices contribute only powers of $\log n$ per local change.

In this overview, we take spectrum differences modulo two. Around an
odd closed walk these differences sum to zero. Read the walk in order,
keeping independent differences and skipping any equal to one
already kept. Stop when a difference is a sum of at least two retained
differences. Such a step must occur: otherwise every step is one of the independent retained vectors,
and their zero sum forces each vector to occur an even number of times,
contradicting the odd length of the walk. Thus some spectrum difference $d$
satisfies $d=g_1+\dots+g_s$ for $s\ge2$ earlier retained differences
$g_1,\ldots,g_s$.

Each local change of $d$ lies within a group already formed by the
changes of the $g_i$. Count the changes of $d$ as well. Every group
must contain at least two changes, since a lone change could not cancel
modulo two. There are $4(s+1)$ changes altogether, so the number of
groups is at most
\[
   \frac{4(s+1)}2=2s+2\le3s.
\]
This is the group bound needed for the coefficient $3$.

To complete the argument, we need bounded witnesses that can be
recovered from their descriptions. Sections~\ref{sec:witness}
and~\ref{sec:resolve} provide the construction and reconstruction
arguments needed to obtain the claimed witness count.

Two words share a label exactly when the matrix sends their spectrum
difference to zero. With $Q$ labels, a nonzero difference is sent to zero with
probability $Q^{-1}$. For a witness of $R$ independent differences,
these events are independent, so all $R$ differences are sent to zero with
probability $Q^{-R}$. Taking
$Q$ of order $n^3$ times a suitable power of $\log n$ makes the union
bound over witnesses small.

For partners reached with fewer than four edits, or with edits close to
one another or to an end, the spectrum difference need not be a sum of
four separate local changes.
Such partners have fewer choices, and we show that few words share a label
with one. Words that share a label with such a partner are the
exceptional words set aside above.
We choose the larger colour class
in each remaining component, then keep the words with a most common label
among the chosen words. This gives a code correcting two deletions with
redundancy about $3\log_2n$.

Linear independence plays a central role in lowering the leading redundancy
coefficient from $2t$ to $2t-1$ in both Varshamov's
argument~\cite{Varshamov1957} and ours. For prime-power $q$, Varshamov
chooses each new parity-check column outside the span of every set of at
most $2t-1$ earlier columns. In our proof, independence of the retained
differences forces the overlaps that lower the coefficient. Our matrix
acts on substring counts, so it is binary for every alphabet and $q$
need not be a prime power.

The two-stage mechanism in our proof is similar to that used in the
incremental synchronization argument of Li, Gabrys and
Farnoud~\cite{LiGabrysFarnoud}: both assign an initial label and then
colour the remaining conflicts between words with the same label.
Their fixed-alphabet codes correct $t$ edits with redundancy
$4t\log_qn+O_{t,q}(\log_q\log_qn)$. Our different way of labelling and
colouring yields the lower redundancy.

\subsection{Organization}

Section~\ref{sec:borrowed} represents words by their substring spectra
and describes the local changes caused by separated edits.
Section~\ref{sec:conflicts} defines the random linear map and the labels,
discards the words whose confusable partners with the same label come from too few or
too close edits, and defines the graph that is two-coloured.
Section~\ref{sec:witness} extracts a witness from any component of this
graph that contains an odd cycle, and Section~\ref{sec:resolve} counts
witnesses by describing their local changes. Section~\ref{sec:main}
combines this count with the probability that all differences of a
witness are sent to zero and selects a code, proving Theorem~\ref{thm:bursts} and hence
Theorem~\ref{thm:intro}.
Section~\ref{sec:related} discusses related work, and
Section~\ref{sec:questions} states the questions that remain.

\section{Substring spectra and local replacements}\label{sec:borrowed}

The argument of Section~\ref{sec:idea} works with differences of
substring counts and describes each by its local changes, one at each
edit. This section makes the local changes precise. Throughout the
proof, fix $t\ge2$ and $b\ge1$. An \emph{edit} now means a
$b$-substring edit, which for $b=1$ is an ordinary insertion, deletion
or substitution, and two words are confusable when some word can be
obtained from each by at most $t$ edits. We first
restrict to a family containing a constant fraction of all words, on
which the substring counts determine the word. We then show that when one
word is transformed into another of the same length by exactly $2t$ edits,
far apart from one another and from the ends, the spectrum difference is a
sum of $2t$ small local changes, one for each edit, which do not interact.

Confusable pairs that fail these conditions, with fewer edits or with
edits too close to one another or to an end, are handled in
Section~\ref{sec:conflicts},
which discards the words involved in such pairs when they share a label.
Every remaining confusable pair with the same label, which we call a
\emph{conflict}, then has a spectrum difference of this form; the
conflicts are the edges of the conflict graph studied there.

\subsection{Substring spectra}

We first restrict to words whose sufficiently long substrings are all
distinct. Each such substring then occurs at most once, so it determines
its position in the word. We use this in two ways: the substring counts
determine the word, and a local change that overlaps one already placed
can be located from their shared letters, which is how overlaps save
positions in Section~\ref{sec:idea}. Fix the word length $n$. Put
\[
   k=2\lceil\log_qn\rceil+2.
\]
For a positive integer $j$, a word is \emph{$j$-unique} if its substrings
of length $j$ are pairwise distinct. Write $\cW_{n,k}$ for the set of
$k$-unique words over $\Sigma$ of length $n$. This restriction retains a
constant fraction of all words, by the following elementary union bound
(see, for example, \cite[Sec.~III]{RepeatFree}).

\begin{lemma}[Most words are $k$-unique]
\label{lem:family}
With $k=2\lceil\log_qn\rceil+2$,
\[
   |\cW_{n,k}|\ \ge\ \Bigl(1-\binom n2 q^{-k}\Bigr)q^n
   \ \ge\ \tfrac78\,q^n\ \ge\ q^n/2.
\]
\end{lemma}

\begin{proof}
Two length-$k$ substrings at distinct prescribed positions agree with
probability $q^{-k}$ in a uniformly random word. If their displacement
is $h<k$, equality leaves $h$ free letters among their $k+h$ positions,
giving $q^h/q^{k+h}=q^{-k}$; disjoint substrings give the same probability.
A union bound over pairs of positions proves the first inequality, and
$\binom n2q^{-k}\le1/(2q^2)\le1/8$ proves the rest.
\end{proof}

Fix an integer $L>k$. A word of length
$L$ is an \emph{$L$-gram}. For a word $y$ let
$R_L(y)\in\mathbb Z^{\Sigma^L}$ be its \emph{spectrum}, whose
coordinate at an $L$-gram $g$ is the number of occurrences of $g$ in $y$
as a substring. Such counts are standard in string matching and in coding
for DNA storage; see Section~\ref{sec:related}.

Consecutive $L$-grams overlap in $L-1$ letters. We record these
overlaps in the de Bruijn graph: its vertices are the words of length
$m=L-1$. Each $L$-gram $a_1\cdots a_L$ labels a directed edge from its prefix
$a_1\cdots a_{L-1}$ to its suffix $a_2\cdots a_L$. A \emph{walk} in
this graph spells a word, and a word of length $h\ge L$ spells a walk
of $h-L+1$ edges. Since $k\le m$, every $k$-unique word is
$m$-unique: two equal substrings of length $m$ at different
positions would begin with equal substrings of length $k$. For such a
word, the next lemma shows that the $L$-grams counted by its spectrum,
viewed as edges, form a simple path in this graph. A simple path can be
followed in only one way, so the spectrum determines the word.
That the spectrum determines an $m$-unique word is due to
Ukkonen~\cite[Theorem~2.2(iii)]{Ukkonen}.
We include the short proof because later sections use the path itself.

\begin{lemma}[The spectrum is a simple path, and determines the word]
\label{lem:path}
Let $y$ be an $m$-unique word of length $n\ge L$. Then $R_L(y)$ is
\emph{Boolean} (all entries are zero or one), and the edges it carries
form a simple directed path
\[
   y_1\cdots y_{L-1}\ \to\ y_2\cdots y_{L}\ \to\ \cdots\ \to\
    y_{n-L+2}\cdots y_n
\]
in the de Bruijn graph, with pairwise distinct vertices. Moreover, if
$z\in\Sigma^n$ satisfies $R_L(z)=R_L(y)$, then $z=y$.
\end{lemma}

\begin{proof}
The vertices listed are the substrings of $y$ of length $m$, pairwise
distinct by $m$-uniqueness, so the walk that $y$ spells is a simple
path, and each of its edges is carried once, which makes $R_L(y)$
Boolean. Along a simple path the initial vertex is the unique vertex of
indegree zero, the terminal vertex is the unique vertex of outdegree
zero, and every other vertex has indegree and outdegree one. A word $z$
of length $n$ with $R_L(z)=R_L(y)$ spells a walk of $n-L+1$ edges using
exactly the edges of that path, each once. It uses the edge leaving the
indegree-zero vertex, and it can be at that vertex only at its start,
since no edge of the path enters it; so the walk begins there. At most one
edge of the path leaves each vertex, so every step of the walk is forced,
and the walk follows the path edge by edge. Hence $z$ spells the same path
and $z=y$.
\end{proof}

\subsection{Shared source letters}

For a suitable choice of $L$, we will show that an edit sufficiently
far from the ends of a $k$-unique word replaces one short path by another
with the same endpoints. In the de Bruijn graph, a \emph{bubble}
consists of two simple directed paths with the same endpoints and no
other common vertex. Replacing
one path by the other removes the first path's grams and introduces the
second path's grams. We orient a bubble by calling these its
\emph{negative} and \emph{positive} paths, respectively.

For example, deleting one zero from $110010$ gives $11010$. With
length-$4$ grams the two paths are
\[
\begin{aligned}
   \text{before deletion:}\quad &110\ \to\ 100\ \to\ 001\ \to\ 010,\\
   \text{after deletion:}\quad &110\ \to\ 101\ \to\ 010.
\end{aligned}
\]
The spectrum difference has entry $-1$ on $1100,1001,0010$, entry
$+1$ on $1101,1010$, and zero elsewhere. The signs identify the two
paths; reversing the edit reverses the signs.

\begin{figure}[htbp]
\centering
\begin{tikzpicture}[x=1.25cm,y=0.9cm,every node/.style={font=\small}]
  \node (s) at (0,0) {$110$};
  \node (a) at (1.5,1.0) {$100$};
  \node (b) at (3.0,1.0) {$001$};
  \node (t) at (4.5,0) {$010$};
  \node (c) at (2.25,-1.0) {$101$};
  \draw[->,thick] (s) -- node[above left] {$1100$} (a);
  \draw[->,thick] (a) -- node[above] {$1001$} (b);
  \draw[->,thick] (b) -- node[above right] {$0010$} (t);
  \draw[->,thick,dashed] (s) -- node[below left] {$1101$} (c);
  \draw[->,thick,dashed] (c) -- node[below right] {$1010$} (t);
  \node at (2.25,1.75) {removed route};
  \node at (2.25,-1.75) {introduced route};
\end{tikzpicture}
\caption{A bubble for one deletion. Deleting one zero from $110010$ to
obtain $11010$ replaces the solid route by the dashed route. The two
routes meet only at their endpoints.}
\label{fig:bubble}
\end{figure}
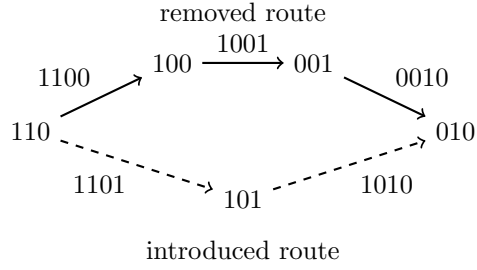

From now on we fix
\begin{equation}
   L=3(k+b+1)
\label{eq:lengths}
\end{equation}
and take $n$ sufficiently large that $k>b$ and $n\ge L$. Positions in a
word start at zero, $z[i:j)$ denotes the
substring at positions $i,\ldots,j-1$, and an occurrence of a
fixed-length substring is called a \emph{window}.

We need three facts about words obtained by one edit from a
\emph{source} word $x\in\cW_{n,k}$: their spectra are still simple paths, the old and new
paths of an edit's local change meet only at their ends
(Lemma~\ref{lem:local-bubble}), and the local changes of edits far
apart are disjoint (Lemma~\ref{lem:coverage}).
With this choice of $L$, all three follow from the next lemma: equal windows of length $m$ in such words come from the
same letters of $x$.

\begin{lemma}[Shared source letters]\label{lem:rigidity}
Let $x\in\cW_{n,k}$, and fix words $y,z$ obtained from $x$ by at most
one edit each. Call the letters of a replacement word \emph{new} and
every other letter \emph{original}; an original letter keeps its
position in $x$, its \emph{source position}. If a length-$m$ window of
$y$ and one of $z$ spell the same word, then at some $k$ consecutive
offsets both windows hold original letters with the same source
positions. Consequently every word within one edit of $x$ is
$m$-unique.
\end{lemma}

\begin{proof}
Compare the windows at equal offsets. Discard offsets occupied by a new
letter on either side, and split the remaining offsets at gaps in the
source positions. Each edit contributes at most one discarded interval,
or one split if its replacement is empty. The two edits, one in each
word, therefore cut the offsets in at most two places: at most $2b$
offsets are discarded and at most three intervals remain. Within each
interval both sequences of source positions are consecutive. Since $m-2b\ge3k$, one
interval contains $k$ consecutive offsets. Their equal letters give two
equal length-$k$ substrings of $x$, whose source positions agree by
$k$-uniqueness.

If equal windows in the same word started at different positions, this
argument would find, at some offset, the same original letter of $x$ in
both windows. That letter would then occur at two positions of the
word, whereas an edit never duplicates an original letter. Thus the
word is $m$-unique.
\end{proof}

\subsection{One edit gives one bubble}

The next lemma handles every edit in the same way: the longest common
prefix and suffix of the words before and after the edit locate the
endpoints of its bubble.

\begin{lemma}[Local bubble]\label{lem:local-bubble}
Let $x$ be $k$-unique and apply one edit, with at least
$2(L+b)$ source letters on either side of the removed substring, or of
the insertion point if nothing is removed. Its spectrum difference is
$R_L(P')-R_L(P)$ for words
\begin{equation}
   P=AUB,\qquad P'=AVB,\qquad U\ne V,\qquad
   |U|,|V|\le b,\qquad m-k-b\le |A|,|B|\le m.
\label{eq:bubble}
\end{equation}
Here $P$ is a substring of $x$ and $P'$ the corresponding substring of
the edited word, the \emph{local words} of the edit; $U$ is the removed
substring, $V$ its replacement, and $A$ and $B$ are the \emph{flanks}.
The local words share their first
and last $m$ letters, and their de Bruijn paths
are simple and meet only at these endpoints. Each local word has length
between $L$ and $2m+b$.
\end{lemma}

\begin{proof}
Write $y$ for the edited word, $p$ for the edit position, $u=|U|$,
$v=|V|$ and $\delta=v-u$. Let $\ell,s$ be the lengths of the longest
common prefix and suffix of $x,y$. Put $c=|x|-s$ and
$c'=|y|-s=c+\delta$. The common prefix and suffix may overlap.

\emph{Locating the paths.}
The unchanged exterior gives $\ell\ge p$ and $s\ge |x|-p-u$.
Moreover,
\[
   \ell-p<k+b,\qquad s-(|x|-p-u)<k+b.
\]
For $\delta=0$ these inequalities follow from $U\ne V$. For
$\delta\ne0$, a common prefix extending $k$ letters beyond
$p+\max\{u,v\}$ would compare equal length-$k$ substrings of $x$
at positions differing by $\delta$, contradicting uniqueness. Reversing
the words gives the suffix inequality.

Set $a=\ell-m$ and take
\[
   P=x[a:c+m),\qquad P'=y[a:c'+m).
\]
The boundary margins make these substrings well defined. They contain the
whole edit and have the form \eqref{eq:bubble}, with flank lengths
\[
   |A|=m-(\ell-p),\qquad
   |B|=m-\bigl(s-(|x|-p-u)\bigr).
\]
They share their first and last $m$ letters. Their length bounds follow
from these flank lengths and $m=3k+3b+2$. Since their exterior walks
are unchanged, the full spectrum difference is $R_L(P')-R_L(P)$.

\emph{Excluding other common vertices.}
Both paths are simple by Lemma~\ref{lem:rigidity}. Consider equal
length-$m$ windows at global starts $i$ in $x$ and $j$ in $y$, with
$a\le i\le c$ and $a\le j\le c'$. Lemma~\ref{lem:rigidity} supplies an
original letter with the same source position at equal offsets in both
windows. If that letter precedes the
edit, then $j=i$ and $i<p\le\ell$. If $i>a=\ell-m$, the windows
include the distinct letters $x_\ell,y_\ell$. Thus this case gives only
$i=j=a$, the initial vertex. If that letter follows the edit,
then $j=i+\delta$ and the windows reach a source position at least
$p+u\ge c$. If $i<c$, they include $x_{c-1},y_{c'-1}$, which differ
by maximality of the common suffix. This case gives only $i=c,j=c'$,
the terminal vertex. The endpoints are distinct because the paths are
simple and nonempty.
\end{proof}

To count conflicts without fixing a source word, we collect every
possible local change of one edit, and every sum of $2t$
vertex-disjoint ones, in a catalogue that depends only on the
parameters.

\begin{definition}[Bubbles and the catalogue]\label{def:catalogue}
Take all pairs \eqref{eq:bubble} sharing their first and last $m$ letters,
of lengths between $L$ and $2m+b$, whose paths are simple and
meet only at their endpoints. Either orientation is a \emph{bubble}.
Its signed vector has entry $-1$ on each edge of the negative path,
$+1$ on each edge of the positive path, and zero elsewhere.
A gram on the negative path is a \emph{negative gram}.

A \emph{rule} is the sum of $2t$ oriented bubbles with pairwise disjoint
vertex sets. The \emph{catalogue} $\cK_0$
is the set of all rules.
\end{definition}

The \emph{support} of a vector $p$ is
$\supp(p)=\{g:p_g\ne0\}$. We view its grams as edges of the de Bruijn
graph, ignoring their directions when discussing connectivity.
Each bubble has connected support; the bubbles of a rule share no
vertex, so they are exactly the $2t$ connected components of its
support. The signs then distinguish their negative and positive paths.
Thus the vector determines its bubbles. Every rule is nonzero with
entries in $\{-1,0,1\}$.

\subsection{Records of separated edits}

Two confusable words are joined by a short sequence of edits
(Lemma~\ref{lem:confuse} below). To locate these edits and later count
them, we describe a transformation of a source word into a
\emph{target} word by a \emph{record}: a list of edits applied one
after another, the first to the source and each later one to the word
produced by the edits before it. Its $i$th edit removes the $u_i\le b$
letters starting at position $p_i$ of the word it is applied to and
inserts in their place a word $V_i$ of length at most $b$, different
from the removed letters. A record is \emph{valid} if each edit fits the
word it is applied to;
together with the source, a valid record determines every intermediate
word and the target.

\begin{lemma}[Confusability gives a short record]\label{lem:confuse}
If two distinct words are confusable, a record of at most $2t$ edits
transforms one into the other.
\end{lemma}

\begin{proof}
Let $u$ be obtained from each word by at most $t$ edits. Apply the edits
from the first word to $u$, then the reverses of the edits from the
second word to $u$, in reverse order; the reverse of an edit is an edit.
\end{proof}

Edits far from one another and from the ends change the spectrum by
disjoint bubbles; we make `far' precise. Put $D_0=8(L+tb)$. A record is
\emph{separated} if its positions
satisfy
\[
   D_0\le p_i\le n-D_0,\qquad |p_i-p_j|\ge D_0\quad(i\ne j).
\]
Separation is only a condition used in the proof; the channel itself
allows overlapping edits. Section~\ref{sec:conflicts} will discard
the exceptional equal-label conflicts that have a shorter or
nonseparated record.

\begin{lemma}[Coverage of separated conflicts]\label{lem:coverage}
Let $x$ be $k$-unique, and let a separated record of exactly $2t$ edits
transform $x$ into a target $z$ of the same length. Then
$R_L(z)-R_L(x)\in\cK_0$.
\end{lemma}

\begin{proof}
We show by induction that each edit acts on letters of $x$ that no
earlier edit has touched. Mark the block written by each earlier edit,
or its insertion point if it wrote nothing. Each mark has length at
most $b$, and earlier edits shift it by at most $2tb$ positions, so by
separation the next edit avoids every mark: it acts on an unchanged
substring of $x$, or an unchanged insertion point, at a source position
$a_i$ with $|a_i-p_i|\le2tb$. Hence $|a_i-a_j|\ge D_0-4tb$ for distinct
edits.

Apply Lemma~\ref{lem:local-bubble} to each edit, placed on $x$ at
$a_i$. Its local word
$P$ lies in $[a_i-m,a_i+u_i+m)$. These intervals exist and are disjoint:
their lengths are at most $2m+b<D_0-4tb$, while
$D_0-2tb-b>2(L+b)$ gives the required boundary margins. If two bubbles
shared a vertex, Lemma~\ref{lem:rigidity} would place the source
position of one original letter in both intervals, a contradiction.

For unchanged exterior words $X,Y$, the common length-$m$ prefixes and
suffixes give
\[
   R_L(XP'Y)-R_L(XPY)=R_L(P')-R_L(P).
\]
Apply this identity to the disjoint replacements. Their sum is the full
spectrum difference, so the difference is a rule.
\end{proof}

\section{Random labels and the conflict graph}\label{sec:conflicts}

This section defines the random binary linear map of Section~\ref{sec:idea}.
Two words have the same label exactly when this map sends their spectrum
difference to zero.
We also discard the words involved in exceptional equal-label
conflicts: those with fewer than $2t$ edits or with edits too close to
one another or to an end. In the remaining conflict graph, the spectrum
difference along every edge is a rule, hence a sum of $2t$ disjoint
bubbles.

We assign each $L$-gram an independent random string of $r$ bits and
sum these strings according to the spectrum counts, with addition
modulo two. The resulting label has $Q=2^r$ possible values. Write
$\bar p$ for the reduction modulo two of an integer vector $p$.
Only the counts are reduced modulo two; the coordinates are indexed by
$\Sigma^L$, so no algebraic structure on $\Sigma$ is needed.

\begin{definition}[Labels and surviving rules]\label{def:hash}
Let $r\ge1$ be an integer and put $Q=2^r$. Draw independent uniform
vectors $\xi_g\in\mathbb F_2^r$ for the $L$-grams $g$. On integer
vectors define the map
\[
   H(p)=\sum_g\bar p_g\xi_g\ \in\mathbb F_2^r.
\]
A word $y$ receives the \emph{label} $H(R_L(y))$.
The set of \emph{surviving rules} is
\[
   \cK=\{w\in\cK_0:H(w)=0\}.
\]
\end{definition}

By linearity,
\begin{equation}
   H(p)=H(p')\quad\Longleftrightarrow\quad H(p'-p)=0.
\label{eq:collision-test}
\end{equation}
The union bound of Section~\ref{sec:main} needs the probability that
several differences are all sent to zero. When their reductions are
independent modulo two, these events are independent:

\begin{lemma}[Independent differences]\label{lem:hash-independent}
Let $w_1,\ldots,w_R$ be integer vectors with linearly independent
reductions over $\mathbb F_2$. The values $H(w_1),\ldots,H(w_R)$ are
independent and uniform on $\mathbb F_2^r$. In particular, for fixed rules with this property,
\[
   \Pr_H[w_1,\ldots,w_R\in\cK]=Q^{-R}.
\]
For any integer vector $w$ with $\bar w\ne0$, the probability that
$H(w)=0$ is $Q^{-1}$.
\end{lemma}

\begin{proof}
Let $M$ be the binary matrix whose rows are the $\bar w_j$. Full row
rank makes $M:\mathbb F_2^{\Sigma^L}\to\mathbb F_2^R$ surjective. For
each of the $r$ output bits, that bit of $\xi_g$ over all $L$-grams $g$
is a uniform vector of $\mathbb F_2^{\Sigma^L}$, so its image under $M$,
which lists that bit of $H(w_1),\ldots,H(w_R)$, is uniform on
$\mathbb F_2^R$. The $r$ output bits use independent
coefficient vectors. Thus all $rR$ bits of
$(H(w_1),\ldots,H(w_R))$ are independent and uniform. The last
statement is the case $R=1$.
\end{proof}

Lemma~\ref{lem:hash-independent} bounds the probability that a single
difference is sent to zero only when the difference is nonzero modulo
two. Lemma~\ref{lem:nonsep} needs this for every distinct partner of
the same length, including partners with repeated substrings:

\begin{lemma}[Parity spectra distinguish equal-length words]\label{lem:parity-path}
Let $x\in\Sigma^n$ be $m$-unique, with $n\ge L$. For every
$z\in\Sigma^n$ distinct from $x$, the difference $R_L(z)-R_L(x)$ is
nonzero modulo two.
\end{lemma}

\begin{proof}
If the parity spectra agreed, each of the $n-L+1$ distinct $L$-grams
of $x$ would occur an odd, positive number of times in $z$. Since $z$
has exactly $n-L+1$ grams counted with multiplicity, each would occur
once and no other gram could occur. The integer spectra would agree,
forcing $z=x$ by Lemma~\ref{lem:path}, a contradiction.
\end{proof}

Lemma~\ref{lem:coverage} does not cover conflicts with fewer than $2t$
edits, or with edits too close to one another or to an end, so their
differences need not be rules, and the counting of
Section~\ref{sec:resolve} does not apply to them. They are cheap to
discard. A partner reached with fewer than $2t$ edits has
fewer positions to choose, and an edit close to another edit or to an end
has only $O(L)$ positions. So such partners are fewer than general
confusable partners by a factor of about $n/L$.

\begin{lemma}[The nonseparated family is thin]\label{lem:nonsep}
Let $\cE$ consist of the words $x\in\cW_{n,k}$ for which a record of
at most $2t$ edits transforms $x$ into a distinct length-$n$ word with
the same label and either has fewer than $2t$ edits or is not
separated. There is a constant
$c_{t,b}$, independent of $q$, such that
\[
   \Pr_H[x\in\cE]\le c_{t,b}q^{2bt}Ln^{2t-1}/Q
   \qquad(x\in\cW_{n,k}).
\]
\end{lemma}

\begin{proof}
A record gives the successive positions, removed lengths and replacement
words. Every intermediate length lies between $n-2tb$ and $n+2tb$,
so we may allow each position to range from $0$ to $n+2tb$, even when
it is not valid for that step. There are $b+1$ removed lengths and
at most $2q^b$ replacement words per edit. Records with $e$
edits therefore number $O_{t,b}(q^{be}n^e)$.

For $e=2t$, failure of separation fixes one position to within
$D_0=8(L+tb)$ of another position or a boundary. After choosing the
other $2t-1$ positions there are only $O_t(L)$ choices for that one,
since $b<L$. Together with records having $e<2t$, this gives at most
$c_{t,b}q^{2bt}Ln^{2t-1}$ exceptional records.

Fix such a record with target $z\ne x$ of length $n$.
Lemma~\ref{lem:parity-path} makes $R_L(z)-R_L(x)$ nonzero modulo two,
even if $z$ has repeated substrings. Thus $H(R_L(z)-R_L(x))=0$ with probability
$Q^{-1}$ by Lemma~\ref{lem:hash-independent}. A union bound over the
records proves the claim.
\end{proof}

We work with confusable pairs that have the same label. The next
sections show that few words lie in components of the graph below that
contain odd cycles; discarding them leaves a bipartite graph.
Fix the map $H$ and the surviving rules $\cK$ from
Definition~\ref{def:hash}.

\begin{definition}[The conflict graph]\label{def:conflictgraph}
The \emph{conflict graph} $\Gamma_H$ has vertex set
$\cW_{n,k}\setminus\cE$, where $\cE$ is the nonseparated family of
Lemma~\ref{lem:nonsep}. Two distinct vertices are joined when they are
confusable and have the same label.
\end{definition}

For every edge $\{y,z\}$ of $\Gamma_H$, in either order,
$R_L(z)-R_L(y)\in\cK$. Indeed, Lemma~\ref{lem:confuse} gives a record
of at most $2t$ edits from $y$ to $z$. Since $y\notin\cE$ and the labels
agree, this record has exactly $2t$ edits and is separated.
Lemma~\ref{lem:coverage} therefore puts the difference in $\cK_0$, and
the equal labels put it in $\cK$ by \eqref{eq:collision-test}.
Both endpoint spectra are Boolean
by Lemma~\ref{lem:path}.

It remains to show that, for some choice of $H$, a constant fraction of
all length-$n$ words lies in bipartite components of $\Gamma_H$.
Choosing a larger side in each such component and then a largest label
class gives a code of size at least a constant multiple of $q^n/Q$.

Walks in $\Gamma_H$ visit words. The overlaps between their spectrum
differences are measured in the de Bruijn graph, whose edges are the
length-$L$ substrings.

\section{Bounded witnesses}\label{sec:witness}

This section will show that every odd closed walk in $\Gamma_H$
starting at a word $x$ yields a small collection of surviving rules whose
reductions are independent and whose bubbles overlap substantially.
We will count candidate collections defined using only $x$ and the
catalogue, independently of $H$; the extracted collection will be one
of these candidates. In
Section~\ref{sec:main}, a union bound will show that, for each fixed $x$,
the probability that any candidate survives is small. This bounds
the expected fraction of words in components containing odd closed
walks. We therefore need bounds on both the number of candidate
collections associated with $x$ and the probability that a fixed
candidate survives.

\subsection{The counting goal and the saving from a relation}

Fix a word $x\in\cW_{n,k}$ and write $v=R_L(x)$. Recall that a rule is
a sum of $2t$ vertex-disjoint bubbles, one for each edit. To count the
collections, we will describe them using positions in $x$ together with
additional data. Each recorded position has at most $n$ possible values,
so a description specifying $c$ positions contributes at most $n^c$
possibilities to the count. We therefore seek descriptions using few
positions. The number of positions needed depends on how the rules'
bubbles overlap; we measure this overlap as follows.

\begin{definition}[Component count]\label{def:adjacency}
For a set $F$ of rules, put $S(F)=\bigcup_{g\in F}\supp(g)$.
Take the de Bruijn subgraph consisting of these edges and their
incident vertices, and ignore edge directions. Write $c(F)$ for the
number of connected components of this subgraph.
\end{definition}

Edges are not cancelled between different rules. Each bubble is
connected, so it lies in one connected component, and bubbles that share
a vertex lie in the same one. If no bubbles from different rules share a
vertex, then $c(F)=2t|F|$; overlaps reduce this number. For example,
when $t=2$, each rule has four disjoint bubbles. Suppose two rules have
exactly one pair of bubbles, one from each rule, that share a vertex.
For the set $F$ of these two rules, $S(F)$ has seven connected
components: the two touching bubbles form one, and the other six remain
separate.

The counting argument in Section~\ref{sec:resolve} will show that
collections $F$ satisfying the reconstruction condition in
Definition~\ref{def:generating} can be described using
just one position in $x$ per connected component of $S(F)$. For
$R=|F|$ and $c=c(F)$, the letters recorded in the description have
$q^{O_{t,b}(R)}$ possibilities; the other choices are polynomial in
$L$ and $R$ per rule. The resulting count is at most
$n^cq^{O_{t,b}(R)}(LR)^{O_{t,b}(R)}$.
If their reductions are independent over $\mathbb F_2$, the probability
that all $R$ rules survive is $Q^{-R}$ by Lemma~\ref{lem:hash-independent}.
We will choose $Q=q^{O_{t,b}(1)}n^{2t-1}L^{O_{t,b}(1)}$. Looking just at the powers of $n$
in the count and survival probability, their product is
\[
   n^c\,n^{-(2t-1)R}=n^{\,c-(2t-1)R}.
\]
This explains the goal: find independent rules with enough overlap
to bring $c$ down to about $(2t-1)R$. We account for the remaining
factors when stating and counting the witnesses.

We keep each rule as a signed integer vector to record its local edits,
and use its reduction modulo two for linear relations and independence.
Since a rule has entries in $\{0,1,-1\}$,
$\supp(\bar g)=\supp(g)$.

Along a closed walk in $\Gamma_H$, the spectrum differences sum to zero.
Read these differences in order, retaining a difference whenever its
reduction modulo two is independent of the previously retained reductions,
and skipping those equal to a retained reduction. An odd closed walk must
eventually produce a reduction that is a sum of at least two retained
reductions. Otherwise, every step reduction would equal one of the retained
independent vectors. In the sum of all step reductions, the coefficient of
each retained vector is its number of occurrences modulo two. Since the
sum is zero, independence forces every coefficient to be zero. Thus each
retained vector occurs an even number of times, so the walk has even
length---a contradiction.

This produces a rule $w$ and a set $I$ of at least two retained rules
satisfying $\bar w=\sum_{g\in I}\bar g$. The next lemma derives the
component bound from this relation without requiring independence. The
independence of the retained reductions will be used separately to bound
the probability that all rules in $I$ survive.

\begin{samepage}
\begin{lemma}[Savings from a relation]\label{lem:saving}
Let $I$ be a finite set of rules. If a rule $w$ satisfies
$\bar w=\sum_{g\in I}\bar g$, then
\[
   c(I)\le t|I|+t.
\]
If $|I|\ge2$, then $c(I)\le(2t-1)|I|$.
\end{lemma}
\end{samepage}

\begin{proof}
The relation gives $\supp(w)\subseteq S(I)$. Every bubble of $w$ or
of a rule in $I$ has connected support, so lies wholly in one connected
component of $S(I)$. Count these bubbles with multiplicity, separately
for every rule, including $w$. Each component contains at least two
occurrences: a single occurrence would have a nonzero reduction there,
contradicting $\bar w+\sum_{g\in I}\bar g=0$. The total number of
occurrences is $2t(|I|+1)$, giving $c(I)\le t(|I|+1)$.
The extra term $t$ accounts for the bubbles of $w$.
When $|I|\ge2$, we have
\[
   (t-1)|I|\ge2(t-1)\ge t,
\]
where the last inequality uses $t\ge2$. Thus
\[
   c(I)\le t|I|+t
   \le t|I|+(t-1)|I|
   =(2t-1)|I|,
\]
proving the second bound.
\end{proof}

When the reductions of the rules in $I$ are independent and $|I|\ge2$,
this is the component bound required by the counting goal. The rule
$w$ is used to force overlaps; it is not included in $I$. For fixed $I$,
the probability that all its rules survive is $Q^{-|I|}$ by
Lemma~\ref{lem:hash-independent}.

\subsection{Retaining the rules needed for reconstruction}\label{sec:retain}

The relation gives the desired component bound for $I$, but the counting
argument also needs an order in which every negative gram of a rule
occurs either in $x$ or in an earlier rule. Section~\ref{sec:resolve} uses this order
to recover the rules, reading the letters of each negative gram either
from $x$ or from a recovered rule with smaller index.

A walk in the conflict graph supplies an order in which every negative
gram occurs in $x$ or in an earlier rule. Along a walk from $x$, retain
the differences whose reductions are independent, in the order in which
they occur. Then a gram that a later step removes but $x$ lacks lies on a path of an
earlier retained rule, because the current spectrum differs from $v$
only on the supports of retained rules; the proof of
Lemma~\ref{lem:witness} gives the details. Either sign is allowed: we
need the gram's letters, not a claim that an earlier retained rule
introduced it.

\begin{definition}[Generating sequences]\label{def:generating}
A \emph{generating sequence at $v=R_L(x)$} is a finite sequence
$G=(g_1,g_2,\ldots)$ of rules whose reductions are linearly independent
over $\mathbb F_2$ such that, writing $N(g_j)=\{h:(g_j)_h=-1\}$ for
the negative grams of $g_j$,
\[
   N(g_j)\subseteq\supp(v)\cup\bigcup_{i<j}\supp(g_i)
   \qquad\text{for every }j.
\]
Thus every negative gram occurs in $x$ or in an earlier rule of the
sequence.
\end{definition}

If a negative gram of $g_j$ is absent from $x$, an earlier rule whose
support contains it is a \emph{supplier} of $g_j$. We call the underlying
set of a generating sequence a \emph{generating set}; which rules survive,
and the number $c(F)$, depend only on this set.

A subset of a generating sequence need not be generating: it may omit an
earlier supplier. The following grouping lets us retain each supplier
along with the rules that need it.

\begin{definition}[Rule blocks]\label{def:blocks}
Let $F$ be a set of rules. Two rules of $F$ are \emph{adjacent} if a bubble of one shares a de Bruijn
vertex with a bubble of the other. Group rules together when they are
linked by a chain of such adjacencies; these groups are the
\emph{blocks} of $F$. Distinct blocks have vertex-disjoint supports, so
$c(F)$ is the sum of $c(B)$ over the blocks $B$.
\end{definition}

The rules in the two-rule example above form one block, although $S(F)$
has seven connected components. Connected components of $S(F)$ group
de Bruijn edges; blocks group whole rules, each of which contains $2t$
disjoint bubbles. Figure~\ref{fig:blocks} shows the distinction for the
$t=2$ example above.

\begin{figure}[htbp]
\centering
\begin{tikzpicture}[x=1cm,y=1cm,every node/.style={font=\small}]
  \node[left] at (-0.45,1.35) {$g_1$};
  \node[left] at (-0.45,-1.35) {$g_2$};

  \draw[dashed,rounded corners] (-0.35,-2.15) rectangle (7.75,2.05);
  \node[anchor=south east] at (7.7,2.05) {one rule block};

  \foreach \x/\lab in {0/1,1.65/2,3.30/3} {
    \draw (\x,1.35) .. controls +(0.35,0.45) and +(-0.35,0.45) .. +(1.05,0);
    \draw (\x,1.35) .. controls +(0.35,-0.45) and +(-0.35,-0.45) .. +(1.05,0);
    \node at (\x+0.525,0.72) {\scriptsize \lab};
  }

  \foreach \x/\lab in {0/5,1.65/6,3.30/7} {
    \draw (\x,-1.35) .. controls +(0.35,0.45) and +(-0.35,0.45) .. +(1.05,0);
    \draw (\x,-1.35) .. controls +(0.35,-0.45) and +(-0.35,-0.45) .. +(1.05,0);
    \node at (\x+0.525,-1.98) {\scriptsize \lab};
  }

  \coordinate (u) at (5.05,1.35);
  \coordinate (v) at (6.05,0);
  \coordinate (w) at (7.05,-1.35);
  \draw (u) .. controls +(0.45,0.20) and +(-0.20,0.50) .. (v);
  \draw (u) .. controls +(0.15,-0.55) and +(-0.55,0.15) .. (v);
  \draw (v) .. controls +(0.55,-0.15) and +(-0.15,0.55) .. (w);
  \draw (v) .. controls +(0.20,-0.50) and +(-0.45,-0.20) .. (w);
  \fill (v) circle (1.2pt);
  \node[right] at (6.10,0.08) {$v$};
  \node at (6.05,-0.55) {\scriptsize 4};
\end{tikzpicture}
\caption{Schematic example for $t=2$. Each row is a rule with four disjoint
bubbles. One bubble of $g_1$ and one of $g_2$ share the vertex $v$, so the two
rules are adjacent and form one block. The numbers label the seven connected
components of $S(\{g_1,g_2\})$: the touching pair gives component $4$, while
the other six bubbles remain separate.}
\label{fig:blocks}
\end{figure}
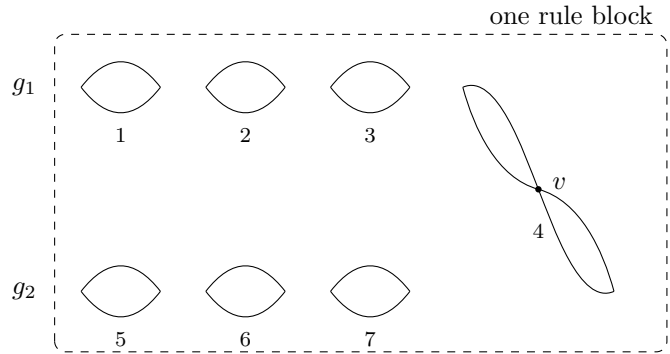

In this subsection, let $G$ be a generating sequence at $v=R_L(x)$,
with each subset of $G$ given its inherited sequence order.

Every union of whole blocks of $G$ is itself generating. Indeed, every
negative gram absent from $x$ occurs in the support of an earlier
rule of $G$. Sharing that gram makes the two rules adjacent, so the
supplier lies in the same block and remains in the subsequence.

\begin{lemma}[Connected blocks]\label{lem:block}
Adjoining a new rule adjacent to at least one rule of a set $F$ increases $c(F)$ by at most
$2t-1$. Consequently a set $B$ connected by rule adjacencies satisfies
\[
   c(B)\le(2t-1)|B|+1.
\]
\end{lemma}

\begin{proof}
One of the new rule's $2t$ bubbles joins an existing connected component
of $S(F)$, so its other $2t-1$ bubbles add at most $2t-1$ components.
For such a set $B$,
start with one rule and adjoin the others along a spanning tree of these
adjacencies.
\end{proof}

Now suppose $I\subseteq G$, $|I|\ge2$, and a rule $w$ satisfies
$\bar w=\sum_{g\in I}\bar g$. Lemma~\ref{lem:saving} gives
$c(I)\le(2t-1)|I|$, but $I$ need not be generating: a supplier of one
of its rules may lie outside it. Let $\mathcal B$ be the blocks of $G$
meeting $I$, and put
\[
   F=\bigcup_{B\in\mathcal B}B,\qquad R=|F|.
\]
Keeping whole blocks retains the suppliers of the rules in $I$, as well
as any suppliers needed by the added rules, because each supplier lies
in the same block as the rule it supplies. Thus $F$ is generating in
the inherited order. Its reductions remain independent because
$F\subseteq G$.

Each block in $\mathcal B$ is connected by rule adjacencies and meets
$I$. Starting from $I$, we can therefore add the remaining rules of
these blocks along adjacency paths, with each new rule adjacent to a
rule already present. Lemma~\ref{lem:block} bounds the increase in the
component count by $2t-1$ per added rule, so
\[
\begin{split}
   c(F)
   &\le c(I)+(2t-1)(R-|I|)\\
   &\le(2t-1)|I|+(2t-1)(R-|I|)\\
   &=(2t-1)R.
\end{split}
\]
The enlargement therefore preserves the component bound. It remains to
bound the size $R$ of the resulting generating set.

\subsection{Extracting a bounded witness}

The enlargement in Section~\ref{sec:retain} gives a generating set with
the desired component bound, but it could retain too many rules if the
blocks are large. We therefore control block sizes while reading an odd
closed walk. Retain a step difference whenever its reduction is
independent of the previously retained reductions, and skip reductions
equal to one already retained. Stop at the first of the following events:
\begin{enumerate}
\item The next step reduction is a sum of at least two retained reductions.
\item Appending a difference whose reduction is independent makes a block
contain more than $L$ rules.
\end{enumerate}
Until one of these events occurs, every retained block has at most $L$
rules. In the first case, the relation will require us to keep at most
$2t$ such blocks, giving at most $2tL$ rules. In the second case, we keep
the newly enlarged block. The appended rule can join only $O_t(L)$
previous blocks, each with at most $L$ rules, so this block has
$O_t(L^2)$ rules. The next lemma proves these size estimates and the
required generating and component-count properties.

\begin{lemma}[Bounded witness]\label{lem:witness}
If the component of $x\in\cW_{n,k}\setminus\cE$ in $\Gamma_H$ is not
bipartite, there is a generating sequence $F\subseteq\cK$ at $R_L(x)$
of $R$ surviving rules such that
\[
   2\le R\le M:=1+2t(2L+2b-2)L,\qquad c(F)\le(2t-1)R+1.
\]
When $R\le L$, the stronger bound $c(F)\le(2t-1)R$ holds.
\end{lemma}

\begin{proof}
Take an odd closed walk in $\Gamma_H$ based at $x$: travel to an odd
cycle in its component, traverse it, and return along the outward
path. Read its signed spectrum differences, all in $\cK$, maintaining
a list $G$ whose reductions modulo two are independent, initially empty.
For the next difference $w$, skip it if $\bar w$ equals a retained
reduction. If $\bar w$ is independent of the retained reductions,
append $w$, stopping if a block now has more than $L$ rules. Otherwise
stop at the new dependence.

\emph{The retained sequence is generating.} Just before an append, the
current spectrum $p$ satisfies
$\bar p-\bar v\in\spn_{\mathbb F_2}\{\bar g:g\in G\}$, since every
earlier difference was retained or equals a retained reduction.
Outside $S(G)$ this gives $p=v$, since both spectra are Boolean.
The new difference $w$ has $N(w)\subseteq\supp(p)$ because the two
endpoint spectra are Boolean. Hence
\[
   N(w)\subseteq\supp(p)\subseteq\supp(v)\cup S(G),
\]
which makes the append generating. Thus every list obtained is generating.

\emph{A stopping event must occur.} If neither event occurs, every step
reduction equals a vector in the final independent retained set. The step reductions
sum to zero around the closed walk, so each retained vector occurs an
even number of times. The walk would then have even length, a contradiction.

\emph{A new dependence.} Every block of $G$ has at most $L$ rules.
Write the unique expansion
\[
   \bar w=\sum_{g\in I}\bar g,\qquad I\subseteq G.
\]
Here $|I|\ge2$: $\bar w$ is nonzero and is not a retained reduction.
Let $\mathcal B$ be the blocks of $G$ meeting $I$ and take
$F=\bigcup_{B\in\mathcal B}B$. By the enlargement argument in
Section~\ref{sec:retain}, $F$ is generating and satisfies
$c(F)\le(2t-1)|F|$. Since each block has at most $L$ rules, it remains
to show that $|\mathcal B|\le2t$.

The supports of distinct blocks are vertex-disjoint. Since
$\bar w=\sum_{g\in I}\bar g$, we have
$\supp(w)\subseteq S(I)\subseteq S(G)$. The $2t$ bubbles of $w$ have
disjoint vertex sets (including endpoints), so their supports partition
$\supp(w)$. Each bubble is connected and hence lies wholly in the support
of one block. For each $B\in\mathcal B$, rules outside $B$ contribute
nothing on $S(B)$, so the restriction of $\bar w$ to the coordinates in $S(B)$
is
\[
   \bar w|_{S(B)}=\sum_{g\in I\cap B}\bar g,
\]
which is nonzero because the reductions in $G$ are independent. Thus
$\supp(w)$ meets $S(B)$,
which therefore contains a bubble of $w$. Since each bubble lies in one
block, distinct blocks contain distinct bubbles of $w$. Consequently
$|\mathcal B|\le2t$ and
\[
   2\le |I|\le R=|F|\le2tL.
\]

\emph{A large block.} Here the construction stopped because a block
acquired more than $L$ rules. Let $g$ be the rule just appended. Before
this adjunction every block had at most $L$ rules; call these the
previous blocks. Take $F$ to be the new block containing $g$. It
consists of $g$ and the previous blocks adjacent to $g$, and it is the
block that exceeded $L$ rules, since appending $g$ left every other
block unchanged. As a whole block of $G$ it is generating, and it is
connected by rule adjacencies. To bound its size, we count the previous
blocks joined by $g$. The previous blocks have vertex-disjoint
supports, so each previous block adjacent to $g$ shares its own vertex
with a bubble of $g$. These shared vertices are distinct for distinct
previous blocks; otherwise the two blocks would have intersecting
supports and would already be one block. Each local word has length
at most $2m+b$, so each path has at most $L+b-1$ edges and $L+b$
vertices. The two paths share their endpoints, so a bubble has at most
$2L+2b-2$ vertices. Hence at most
$2t(2L+2b-2)$ previous blocks, each with at most $L$ rules, are adjacent to
$g$. Counting $g$ itself as well, we obtain
\[
   L<R=|F|\le1+2t(2L+2b-2)L=M.
\]
Lemma~\ref{lem:block} gives $c(F)\le(2t-1)R+1$, as required.
\end{proof}

A \emph{witness} is a generating sequence $F$ satisfying the bounds of
Lemma~\ref{lem:witness} on its size $R$ and on $c(F)$. It need not be
connected by rule adjacencies. It \emph{survives} when all its rules
belong to $\cK$.

The large-block case allows one extra connected component, which
contributes an additional factor of at most $n$ to the count in
Section~\ref{sec:resolve}. In this case $R>L$, and $L\ge\log_q n$ gives
\[
   n\le q^L<q^R.
\]
Thus one extra factor of $q$ per rule covers this cost. In both cases,
$R\le M=O_t(L^2)$ because $b<L$, so the remaining description choices,
which contribute $(LR)^{O_{t,b}(R)}$, are bounded by $L^{O_{t,b}(R)}$.
Section~\ref{sec:main} absorbs these factors and the recorded letters
through suitable powers of $q$ and $L$ in the choice of $Q$.

\section{Counting generating sets}\label{sec:resolve}

This section proves the count behind Section~\ref{sec:idea}: a
generating sequence needs only one freely chosen position in $x$ for
each connected component formed by its bubbles. Two bubbles overlap
when their paths share a de Bruijn vertex, that is, an $m$-letter
substring. A shared vertex supplies the known letters used in the
reconstruction.

We start the recovery of each connected component by recording in $x$
the position of one bubble's negative path. An overlapping bubble then inherits a known
substring of at least $k$ letters. To extend this substring, we use
the length-$L$ windows of its negative path.
By the generating-sequence property, each window occurs either in $x$,
where $k$-uniqueness determines its position, or on a path of a bubble
of an earlier rule. If that bubble is not yet known, we recover it first, starting
from the known letters of the window, which also lie on its path. Each
such detour moves to an earlier rule, so the process ends.

The proof below makes this recovery precise. After the initial positions
in $x$, all choices name bubbles, short offsets and a bounded number of
letters per bubble. We count the letter choices separately.

\begin{lemma}[Counting generating sets]\label{lem:resolve}
There is an integer $a_{t,b}\ge1$, depending only on $t,b$, such that for
every $x\in\cW_{n,k}$ and integers $R\ge1$ and $0\le c\le2tR$, the
number of generating sets $G$ at $R_L(x)$ with $R$ rules and $c(G)=c$
is at most
\[
   n^cq^{6btR}(LR)^{a_{t,b}R}.
\]
\end{lemma}

\begin{proof}
We count generating sets by describing how to reconstruct them from
the fixed word $x$ and a short description. We
will show that every generating set admits such a description and that
each description determines at most one underlying set of rules. Thus
it suffices to count the possible descriptions. The reconstruction
starts with one bubble in each connected component of $S(G)$ whose
negative path occurs in $x$; its position in $x$ supplies a starting
point for recovering the rest of that component.

\emph{Starting each connected component of $S(G)$.} Put $\nu=2tR$. Index the bubble occurrences by
$(i,j)$, where $i$ is the rule's position in the sequence and
$1\le j\le2t$; identical bubbles in different rules retain separate
indices. Each connected component of $S(G)$ contains a bubble whose
negative path spells a substring of $x$: choose one whose rule is
earliest in the sequence, and call it the \emph{root} of its component.
A negative gram of the root absent from $x$ would occur on a path of an earlier
rule in the same connected component of $S(G)$, a contradiction. All its negative grams
therefore occur in $x$, so its negative path is a subpath of the simple
path $R_L(x)$ from Lemma~\ref{lem:path}.

\emph{Headers and windows.} Write a bubble's negative and positive
local words as $AUB$ and $AVB$. Its \emph{header} records
$|A|,|B|$ and the ordered pair $U,V$. Each flank length has at most
$L$ choices, and each of $U,V$ has at most $2q^b$ choices, giving at
most $4L^2q^{2b}$ headers. The negative path and header determine the
positive path by replacing $U$ with $V$ at offset $|A|$.

A negative path spells a word of length $L+D$ with
$0\le D\le L+b-2$. Select its first, middle, and
last $L$-gram windows, at the offsets
\[
   0,\quad \lfloor D/2\rfloor,\quad D.
\]
Successive offsets differ by at most
\[
   \lceil D/2\rceil\le\lceil(L+b-2)/2\rceil<L-k,
\]
since $L=3(k+b+1)$. Thus these at most three windows cover the path,
overlap in at least $k$ letters, and every $k$-letter substring lies in
one of them.
In a generating sequence, each selected window either occurs in $x$
or lies on a path of a bubble of an earlier rule. A window in $x$ is determined by
any known $k$-letter substring at a known offset inside it, because $x$ is
$k$-unique.

\emph{Transferring a known substring.} The known substring may lie on a
bubble's positive path, whereas the generating-sequence property supplies
windows of its negative path. We therefore need to transfer a known
$k$-letter substring between the two paths.
To transfer from $AUB$ to $AVB$, remove the known substring's letters
in $U$. Since $k>b$, a nonempty consecutive substring remains in
$AB$, with at least $k-b$ letters. Extend it to $k$ letters using
at most $b$ additional letters: if it starts at offset $a$ in $AB$,
use offset $\min\{a,|AB|-k\}$ for the extended substring.
The resulting interval contains the whole remaining substring;
its position and the required extension offsets are determined by
the header and incoming offset. If the extended substring contains letters on both sides of the
insertion position $|A|$, insert the known block $V$ and retain the
first $k$ letters. Otherwise transfer it unchanged, shifting its start
by $|V|$ if it starts at or after $|A|$.
The reverse transfer interchanges $U,V$.

Thus a known substring on either path, its offset and at most $b$ letters
determine a known $k$-letter substring of the negative path.

\emph{Recovering one bubble.} Let $\beta$ be a bubble that is not yet
recovered, and suppose a $k$-letter substring of one of its paths is
known at a known offset. Transfer it to the negative path as above. Then
fill the selected windows of $\beta$ one at a time: first a window
containing the transferred substring, then each remaining window once
it overlaps a filled one, which gives at least $k$ known letters inside
it. A window in $x$ is filled by the unique lookup. Otherwise it lies on
a path of a bubble $\beta'$ of an earlier rule. If $\beta'$ is not yet
recovered, recover it first by the same procedure, starting from the
known letters of the window at the corresponding offset on its path;
then copy the window from $\beta'$. Each nested recovery concerns a
strictly earlier rule, so the nesting has depth at most $R$, and the
bubbles whose recovery is in progress all belong to rules later than
$\beta'$. Hence $\beta'$ is either recovered or untouched when it is
needed. Once all windows of $\beta$ are filled, its negative path is
known, and the header gives its positive path.

\emph{The order of recovery.} The roots are recovered from their
positions in $x$ and their headers. While some bubble is not recovered,
choose one that shares a vertex with a recovered bubble, and recover it
starting from the $m\ge k$ letters of that vertex. Such a bubble
exists because each connected component of $S(G)$ contains a recovered
root.

\emph{Descriptions.} A description records the headers of all bubbles,
the indices of the $c$ roots with their positions in $x$, and the
choices made by the procedure, in order: the two bubbles, paths and
offsets for each recovery started from a shared vertex; at most $b$
extra letters for each transfer; and, for each window, which window it
is and either that it lies in $x$ or the supplying bubble, path and
offset. Each bubble is recovered once and has at most three selected
windows, so a description lists at most $4\nu$ steps. Given $x$, a
description determines every step of the procedure. We reject
descriptions that call for missing letters, failed lookups in $x$,
invalid offsets or inconsistent letters, so each description determines
at most one underlying set of rules, without using $H$. Every generating
set has a description; we count all descriptions, and those that do not
recover a generating set only enlarge the upper bound.

\emph{Counting.} Each connected component of $S(G)$ contributes one
root position in $x$, so these positions cost $n^c$. The $\nu$ headers
contribute at most $q^{2b\nu}$ times a polynomial factor in $L$ per
bubble. Each bubble is recovered at most once and needs at most $b$
extra letters, contributing at most $q^{b\nu}$ further choices. Root
indices and all other recorded data name only bubbles, paths and
offsets, with polynomially many choices in $L$ and $\nu$. Every later
lookup in $x$ is determined by known letters, so no further position is
recorded. With at most $4\nu$ steps the total number of descriptions is
\[
   n^cq^{3b\nu}(L\nu)^{O_{t,b}(\nu)}\le n^cq^{6btR}(LR)^{a_{t,b}R}
\]
for an integer $a_{t,b}$ depending only on $t,b$, since $\nu=2tR$, $k<L$ and
$LR\ge2$. Every generating set is recovered from one of these
descriptions, proving the bound.
\end{proof}

\section{Proof of the main theorem}\label{sec:main}

This section completes the union bound of Section~\ref{sec:idea}. We
combine the generating-set count with the probability that a fixed set
survives. This bounds the fraction of words in components of $\Gamma_H$
containing odd cycles; the remaining components of $\Gamma_H$ then yield
a large code.

\subsection{Surviving witnesses are rare}\label{sec:count}

Let $a_{t,b}$ be the constant in Lemma~\ref{lem:resolve}. Choose a fixed
integer $\lambda_{t,b}\ge2$, depending only on $t,b$, large enough that, for all
sufficiently large $L$,
\begin{equation}
   (2tR+1)(LR)^{a_{t,b}R}\le L^{\lambda_{t,b}R}
   \qquad(2\le R\le M).
\label{eq:record-cost}
\end{equation}
The factor $2tR+1$ counts the possible values of $c(F)$; the alphabet
factors will be included explicitly below. Such a choice is possible
because each factor on the left is at most $L^{O_{t,b}(R)}$: $b<L$
gives $R\le M\le(10t+1)L^2$, and $2tR+1\le(2t+1)^R$. Also, $L$ tends to
infinity with $n$ for fixed $q$. The required threshold for $L$ depends
only on $t,b$.

\begin{proposition}[Odd cycles are rare]\label{prop:obstruction}
If
\begin{equation}
   Q\ge4q^{6bt+1}n^{2t-1}L^{\lambda_{t,b}},
\label{eq:labels}
\end{equation}
then, for every $x\in\cW_{n,k}$, the probability that $x\notin\cE$ and
its component in $\Gamma_H$ is not bipartite is at most $1/12$.
\end{proposition}

\begin{proof}
The witnesses at a fixed $x\in\cW_{n,k}$ depend only on the catalogue;
$H$ determines which ones survive. We bound this probability without
conditioning on $\Gamma_H$ or $\cE$, both of which depend on $H$.
For each $2\le R\le M$, sum the bound in
Lemma~\ref{lem:resolve} over $0\le c\le(2t-1)R$, allowing also
$c=(2t-1)R+1$ when $R>L$, in which case $n\le q^R$ because
$L\ge\log_qn$. The number of underlying witness sets of size $R$ is at
most
\[
   q^{(6bt+1)R}(2tR+1)n^{(2t-1)R}(LR)^{a_{t,b}R}
   \le \bigl(q^{6bt+1}n^{2t-1}L^{\lambda_{t,b}}\bigr)^R.
\]
The reductions of the rules in a witness are linearly independent, so
each witness set survives with probability $Q^{-R}$ by
Lemma~\ref{lem:hash-independent}. Under \eqref{eq:labels}, the expected
number of surviving witness sets of size $R$ is therefore at most
$4^{-R}$. If $x\notin\cE$ and its component is not bipartite,
Lemma~\ref{lem:witness} supplies a surviving witness, so the
probability in the statement is at most the expected number of
surviving witnesses:
\[
   \sum_{R=2}^{M}4^{-R}
   \le\sum_{R=2}^{\infty}4^{-R}=\frac1{12}.
\]
\end{proof}

\subsection{Selecting the code}

We now combine Lemma~\ref{lem:nonsep} with
Proposition~\ref{prop:obstruction} and select one label class from the
remaining bipartite components of $\Gamma_H$.

\begin{theorem}\label{thm:main}
For fixed $t\ge2$, $b\ge1$ and $q\ge2$, and all sufficiently large $n$
depending on these parameters, there is a code $\cC\subseteq\Sigma^n$
correcting at most $t$ $b$-substring edits with
\[
 \red(\cC)\le(2t-1)\log_qn+O_{t,b}(\log_q\log_qn).
\]
\end{theorem}

\begin{proof}
Choose
\[
   r=\left\lceil\log_2\bigl(4q^{6bt+1}n^{2t-1}L^{\lambda_{t,b}}\bigr)\right\rceil,
   \qquad Q=2^r,
\]
so $4q^{6bt+1}n^{2t-1}L^{\lambda_{t,b}}\le Q<8q^{6bt+1}n^{2t-1}L^{\lambda_{t,b}}$.
Draw $H$ as in Definition~\ref{def:hash}. Discard a word
$x\in\cW_{n,k}$ if it belongs to the nonseparated family $\cE$ of
Lemma~\ref{lem:nonsep}, or if $x\notin\cE$ and its component in
$\Gamma_H$ is not bipartite. The first event has probability at most
\[
   \frac{c_{t,b}q^{2bt}Ln^{2t-1}}{Q}\le\frac{c_{t,b}}{4}q^{-4bt-1}L^{1-\lambda_{t,b}},
\]
which tends to zero and is at most $1/4$ for large $n$, since
$\lambda_{t,b}\ge2$. The second event has probability at most $1/12$ by
Proposition~\ref{prop:obstruction}. The expected fraction of discarded
words is thus at most $1/3<1/2$.

Fix $H$ for which at least half of $\cW_{n,k}$ remains after the
discards, and let $\cC_0$ be the remaining words. Lemma~\ref{lem:family} gives
$|\cC_0|\ge q^n/4$. These words form a union of bipartite components
of $\Gamma_H$. Choose the larger side of the bipartition of each such
component of $\Gamma_H$.
Their union is an independent set $\mathcal I$ with
$|\mathcal I|\ge|\cC_0|/2\ge q^n/8$. Choose the largest label class in $\mathcal I$,
\[
   \cC=\{x\in\mathcal I:\ H(R_L(x))=a\}
   \qquad\text{for some }a\in\mathbb F_2^r.
\]
Then $|\cC|\ge q^n/(8Q)$. Two distinct confusable words of $\cC$
would have equal labels and hence be adjacent in $\Gamma_H$,
contradicting the independence of $\mathcal I$. Thus $\cC$ corrects at
most $t$ $b$-substring edits, and
\[
   \red(\cC)\le\log_q(8Q)
   \le(2t-1)\log_qn+\lambda_{t,b}\log_qL+6bt+1+\log_q64.
\]
For $\log_qn\ge1$ we have $L\le(3b+21)\log_qn$, so this is at most
\[
   (2t-1)\log_qn+\lambda_{t,b}\log_q\log_qn+D_{t,b},
\]
where $D_{t,b}=6bt+7+\lambda_{t,b}\log_2(3b+21)$ is independent of $q$. For each fixed
$q$, increasing the threshold to require $n\ge q^q$ gives
$\log_q\log_qn\ge1$ and absorbs $D_{t,b}$ into the stated remainder.
This proves Theorem~\ref{thm:bursts}. Taking $b=1$ proves
Theorem~\ref{thm:intro} for ordinary mixed edits.
\end{proof}

The $O_{t,b}(\log_q\log_qn)$ term comes from the description costs other
than positions in $x$, which are polynomial in $L$ per rule because
witnesses have $O_t(L^2)$ rules.

\section{Related work}\label{sec:related}

For deletion correction, Alon, Bourla, Graham, He and
Kravitz~\cite{AlonEtAl} improved the greedy bound by a doubly logarithmic
term, explicitly noting that their result extends to every fixed
alphabet. They proved the existence of codes $\cC$ correcting $t$ deletions with
\begin{equation}
   \red(\cC)\ \le\ 2t\log_qn-\log_q\log_qn+O_{t,q}(1).
\label{eq:alon}
\end{equation}
Kong, Tamo and Wei~\cite{KongTamoWei} refine the
constant in this bound and treat further alphabet and distance regimes;
for fixed $q,t$ the leading coefficient remains $2t$.
The argument of Alon et al.\ shows that the graph of pairs confusable under $t$ deletions
has few triangles
and applies a bound of Ajtai, Koml\'os and Szemer\'edi~\cite{AKS} for such
graphs, which gains a logarithmic factor over the greedy choice. A bound
that uses only the degree and the number of triangles cannot gain more,
since sparse random graphs of the same degree also have few triangles and
independence number of order $(N/d)\log d$ for $N$ vertices of
degree~$d$~\cite{Frieze}.

Explicit binary codes correcting $t$ deletions have larger leading
coefficients, except for $t=2$, where Guruswami and
H\aa stad~\cite{GuruswamiHastad} match the coefficient $4$ of
\eqref{eq:lev}; for $t\ge3$, see
\cite{BGZ,SimaGabrysBruck,SongEtAl}. Guruswami and H\aa stad also give
explicit two-deletion codes with redundancy $3\log_2n+O(\log_2\log_2n)$
that decode to a list of at most two words; the codes here decode
uniquely. For fixed $q>2$, Liu, Tjuawinata and Xing~\cite{LiuTjuawinataXing}
construct two-deletion codes with redundancy
$5\log_qn+O_q(\log_q\log_qn)$, in $q$-ary symbols. Song and
Cai~\cite{SongCai} obtain the same leading coefficient for even $q$.

For mixed errors, Spee and Weber~\cite{SpeeWeber} study $q$-ary codes
with separate budgets of $t_1$ insertions/deletions and $t_2$
substitutions. Their existence bound has leading redundancy coefficient
$2(t_1+t_2)$. Taking $t=t_1+t_2$, our code corrects every such budget pair, and also
allows the mixture to vary subject only to the total budget $t$.

In the binary setting, Cullina, Kulkarni and Kiyavash~\cite{CullinaEtAl}
give a broader precedent
for the way our proof first restricts the words to simplify the conflict
graph and then selects an independent set. Their restriction, playing a
role analogous to our label, keeps the words whose Hamming weights have
one fixed residue modulo $t+1$: words confusable under $t$ deletions
have weights differing by at most $t$, so all remaining conflicts are
between words of equal weight. They then find an independent set in
each constant-weight subgraph and take the union.

Precoding followed by syndrome compression gives another way to reduce
the conflicts before separating them. Song, Polyanskii, Cai and
He~\cite{SongEtAl} first encode messages into an auxiliary code with fewer
confusable partners per word, then use a compressed syndrome to separate
confusable words in that family. Their resulting systematic codes have
efficient encoding and decoding.

Li, Gabrys and Farnoud~\cite[Section~V-C]{LiGabrysFarnoud} cite this
precoding method as inspiration for a construction that fixes a checksum
and colours the remaining conflicts. The closest precedent for our proof
is their incremental synchronization
argument~\cite[Theorem~5]{LiGabrysFarnoud}, stated for binary words,
which uses the same two stages. It first fixes a syndrome that corrects fewer edits than the final code
corrects, which lowers the maximum degree of the remaining conflict
graph, and then colours that graph. The number of added bits grows with the logarithm of this degree.
Their fixed-alphabet codes~\cite[Corollary~17]{LiGabrysFarnoud} correct
$t$ edits with redundancy $4t\log_qn+O_{t,q}(\log_q\log_qn)$ and have
polynomial-time encoding and decoding.

In our proof, the label comes from a random linear map on
substring counts, and the witness count shows that, after a small
exceptional family is discarded, the equal-label conflict graph is
bipartite. Choosing a larger side then costs only a factor of two, so
the colouring adds only $\log_q2$ redundant symbols to the label. This
gives $(2t-1)\log_qn+O_t(\log_q\log_qn)$. Whether codes with this
coefficient can be constructed explicitly is open; see
Section~\ref{sec:questions}.

For the $b$-substring-edit model of Theorem~\ref{thm:bursts},
greedy counting gives coefficient $2t$. For binary alphabets and fixed $b$,
Li, Gabrys and Farnoud~\cite{LiGabrysFarnoud} give explicit codes with
redundancy $4t\log_2n+O_{t,b}(\log_2\log_2n)$.

Theorem~\ref{thm:bursts} also covers the special case of fixed-length deletion
bursts. For two such bursts of length $b>1$, Ye et al.~\cite{YeEtAlBursts}
give an existence bound with redundancy at most $4\log_qn+2b$ and
explicit codes with redundancy $5\log_qn+O_{b,q}(\log_q\log_qn)$.
The theorem gives coefficient $3$. The lower bound of Wang et
al.~\cite{WangKongYaakobiDuman} has coefficient $t$ for fixed $t,b,q$.

Counts of substrings of a fixed length are a standard representation of
strings. In string matching they define substring-count distances, which
bound the edit distance from below~\cite{Ukkonen}. In coding for DNA
storage, codes are designed so that sequences are distinguished by these
counts, called profile vectors~\cite{KiahPuleoMilenkovic}, and strings are
reconstructed from their substring spectrum~\cite{MarcovichYaakobi}. Here
the counts serve as the vector on which a random linear map acts, and the
fact that an edit changes them only near its place is what lets
overlapping local changes save positions.

Throughout, the errors are worst-case and their total number $t$ is fixed.
A different line of work studies channels that delete each bit
independently with a fixed probability, where the central question is
the capacity; see the surveys~\cite{Mitzenmacher,CheraghchiRibeiro}.
When the number of deletions grows linearly with $n$, the question
becomes the trade-off between the rate and the fraction of deletions
corrected; see~\cite{GuruswamiWang} for binary codes in the high-rate
regime and the survey~\cite{HaeuplerShahrasbi}.

\section{Further questions}\label{sec:questions}

Three questions remain. First, can the correction term be bounded
independently of $n$? More precisely, does
\[
 \red^*(n)\le(2t-1)\log_qn+O_t(1)
\]
hold? The correction in the present upper bound comes from the factor
$L^{\lambda_{t,1}}$, which pays for the letters and bubble indices
recorded in the descriptions of Section~\ref{sec:resolve}. These costs are polynomial per rule, and their logarithm gives
the $O_t(\log_q\log_qn)$ remainder. The current estimate for nonseparated
conflicts in Lemma~\ref{lem:nonsep} also carries a factor $L$. Sharper
estimates for these contributions, using the structure of bubble
vectors and edit records, could improve the correction term.

Second, what is the sharp leading-order redundancy? The sphere-packing
bound has coefficient $t$ and Theorem~\ref{thm:intro} has coefficient
$2t-1$. Thus
\[
 t\le\liminf_{n\to\infty}\frac{\red^*(n)}{\log_qn}
 \le\limsup_{n\to\infty}\frac{\red^*(n)}{\log_qn}\le2t-1.
\]
Determining these limits, including whether they agree, remains open for
every fixed $q\ge2$ and $t\ge2$.

The present approach suggests possible ways to narrow this gap.
For a set $F$ of $R$ rules connected by rule adjacencies,
Lemma~\ref{lem:block} gives $c(F)\le(2t-1)R+1$, and linear relations
provide further savings (Lemma~\ref{lem:saving}). Stronger structural counting estimates or a
colouring argument allowing more than two colours are possible directions.
A smaller coefficient would also have to overcome a constraint that
involves neither the number of colours nor Lemma~\ref{lem:block}:
Lemma~\ref{lem:nonsep} controls nonseparated conflicts by the upper bound
$c_{t,1}q^{2t}Ln^{2t-1}/Q$. If $Q=n^{\alpha+o(1)}$ with
$\alpha<2t-1$, this bound diverges and no longer shows that only a
small fraction of words in $\cW_{n,k}$ is discarded. Improving the leading
coefficient within this argument therefore also requires a sharper
treatment of nonseparated conflicts, regardless of improvements to the
colouring or to the bound on $c(F)$.

In the Hamming metric, binary BCH codes reach the coefficient $t$ by
using an algebraically structured parity-check matrix in place of
Varshamov's greedy one. Our linear map acts on substring counts and is
chosen at random. A structured choice of this map, or of another linear
representation of the words, is a possible route below $2t-1$ for
deletions.

Third, are there explicit codes with the same leading coefficient? The
code of Theorem~\ref{thm:main} is chosen at random, and even for a fixed
map $H$, choosing a side of every component is not known to be efficient.
For two binary deletions, the list-two codes of Guruswami and
H\aa stad~\cite{GuruswamiHastad} already have redundancy
$3\log_2n+O(\log_2\log_2n)$. The pairs of codewords that their checks do
not separate form a graph. If discarding a small fraction of words made
this graph bipartite, with a colouring that is efficient to compute, one
extra bit could give explicit uniquely decodable codes with coefficient
$3$. Whether this holds is open. An explicit construction would also
require efficient encoding into the retained family without increasing
the leading redundancy coefficient.

\section*{Acknowledgements}

The proofs and the presentation were developed in interaction with large
language models.
The author takes full responsibility for the paper.

\clearpage

\begin{thebibliography}{99}

\bibitem{AbdelGhaffarFerreiraCheng}
K.~A.~S. Abdel-Ghaffar, H.~C. Ferreira, and L.~Cheng,
``Correcting deletions using linear and cyclic codes,''
\emph{IEEE Transactions on Information Theory}, vol.~56, no.~10,
pp.~5223--5234, 2010. doi:10.1109/TIT.2010.2059790.

\bibitem{AKS}
M.~Ajtai, J.~Koml\'os, and E.~Szemer\'edi,
``A note on Ramsey numbers,''
\emph{Journal of Combinatorial Theory, Series A}, vol.~29, no.~3,
pp.~354--360, 1980.

\bibitem{AlonEtAl}
N.~Alon, G.~Bourla, B.~Graham, X.~He, and N.~Kravitz,
``Logarithmically larger deletion codes of all distances,''
\emph{IEEE Transactions on Information Theory}, vol.~70, no.~1,
pp.~125--130, 2024. doi:10.1109/TIT.2023.3304565.

\bibitem{BoseRayChaudhuri1960}
R.~C.~Bose and D.~K.~Ray-Chaudhuri,
``On a class of error correcting binary group codes,''
\emph{Information and Control}, vol.~3, no.~1,
pp.~68--79, 1960. doi:10.1016/S0019-9958(60)90287-4.

\bibitem{BGZ}
J.~Brakensiek, V.~Guruswami, and S.~Zbarsky,
``Efficient low-redundancy codes for correcting multiple deletions,''
\emph{IEEE Transactions on Information Theory}, vol.~64, no.~5,
pp.~3403--3410, 2018.

\bibitem{CheraghchiRibeiro}
M.~Cheraghchi and J.~Ribeiro,
``An overview of capacity results for synchronization channels,''
\emph{IEEE Transactions on Information Theory}, vol.~67, no.~6,
pp.~3207--3232, 2021.

\bibitem{CullinaEtAl}
D.~Cullina, A.~A.~Kulkarni, and N.~Kiyavash,
``A coloring approach to constructing deletion correcting codes from
constant weight subgraphs,''
in \emph{Proc.\ IEEE International Symposium on Information Theory (ISIT)},
2012, pp.~513--517. doi:10.1109/ISIT.2012.6284242.

\bibitem{RepeatFree}
O.~Elishco, R.~Gabrys, E.~Yaakobi, and M.~M\'edard,
``Repeat-free codes,''
\emph{IEEE Transactions on Information Theory}, vol.~67, no.~9,
pp.~5749--5764, 2021. doi:10.1109/TIT.2021.3092056.

\bibitem{Frieze}
A.~M.~Frieze,
``On the independence number of random graphs,''
\emph{Discrete Mathematics}, vol.~81, no.~2, pp.~171--175, 1990.

\bibitem{Gilbert1952}
E.~N.~Gilbert,
``A comparison of signalling alphabets,''
\emph{Bell System Technical Journal}, vol.~31, no.~3,
pp.~504--522, 1952. doi:10.1002/j.1538-7305.1952.tb01393.x.

\bibitem{GuruswamiHastad}
V.~Guruswami and J.~H\aa stad,
``Explicit two-deletion codes with redundancy matching the existential bound,''
\emph{IEEE Transactions on Information Theory}, vol.~67, no.~10,
pp.~6384--6394, 2021.

\bibitem{GuruswamiWang}
V.~Guruswami and C.~Wang,
``Deletion codes in the high-noise and high-rate regimes,''
\emph{IEEE Transactions on Information Theory}, vol.~63, no.~4,
pp.~1961--1970, 2017.

\bibitem{HaeuplerShahrasbi}
B.~Haeupler and A.~Shahrasbi,
``Synchronization strings and codes for insertions and deletions---a
survey,''
\emph{IEEE Transactions on Information Theory}, vol.~67, no.~6,
pp.~3190--3206, 2021.

\bibitem{Hocquenghem1959}
A.~Hocquenghem,
``Codes correcteurs d'erreurs,''
\emph{Chiffres}, vol.~2, pp.~147--156, 1959.

\bibitem{KiahPuleoMilenkovic}
H.~M.~Kiah, G.~J.~Puleo, and O.~Milenkovic,
``Codes for DNA sequence profiles,''
\emph{IEEE Transactions on Information Theory}, vol.~62, no.~6,
pp.~3125--3146, 2016.

\bibitem{KongTamoWei}
X.~Kong, I.~Tamo, and H.~Wei,
``Combinatorial alphabet-dependent bounds for insdel codes,''
\emph{IEEE Transactions on Information Theory}, vol.~71, no.~5,
pp.~3544--3559, 2025. doi:10.1109/TIT.2025.3545061.

\bibitem{Levenshtein1966}
V.~I.~Levenshtein,
``Binary codes capable of correcting deletions, insertions, and reversals,''
\emph{Soviet Physics Doklady}, vol.~10, no.~8, pp.~707--710, 1966.
Russian original: \emph{Dokl. Akad. Nauk SSSR}, vol.~163, no.~4,
pp.~845--848, 1965.

\bibitem{LiGabrysFarnoud}
Y.~Li, R.~Gabrys, and F.~Farnoud,
``Constructing low-redundancy codes via distributed graph coloring,''
\emph{IEEE Transactions on Information Theory}, vol.~72, no.~10,
pp.~7344--7358, 2026. doi:10.1109/TIT.2026.3718236.

\bibitem{LiuTjuawinataXing}
S.~Liu, I.~Tjuawinata, and C.~Xing,
``Explicit construction of $q$-ary 2-deletion correcting codes with low redundancy,''
\emph{IEEE Transactions on Information Theory}, vol.~70, no.~6,
pp.~4093--4101, 2024. doi:10.1109/TIT.2024.3360964.

\bibitem{MarcovichYaakobi}
S.~Marcovich and E.~Yaakobi,
``Reconstruction of strings from their substrings spectrum,''
\emph{IEEE Transactions on Information Theory}, vol.~67, no.~7,
pp.~4369--4384, 2021.

\bibitem{Mitzenmacher}
M.~Mitzenmacher,
``A survey of results for deletion channels and related synchronization
channels,''
\emph{Probability Surveys}, vol.~6, pp.~1--33, 2009.

\bibitem{SimaGabrysBruck}
J.~Sima, R.~Gabrys, and J.~Bruck,
``Optimal systematic $t$-deletion correcting codes,''
in \emph{Proc.\ IEEE International Symposium on Information Theory (ISIT)},
2020, pp.~769--774.

\bibitem{SongCai}
W.~Song and K.~Cai,
``Non-binary two-deletion correcting codes and burst-deletion correcting codes,''
\emph{IEEE Transactions on Information Theory}, vol.~69, no.~10,
pp.~6470--6484, 2023.

\bibitem{SongEtAl}
W.~Song, N.~Polyanskii, K.~Cai, and X.~He,
``Systematic codes correcting multiple-deletion and multiple-substitution errors,''
\emph{IEEE Transactions on Information Theory}, vol.~68, no.~10,
pp.~6402--6416, 2022.

\bibitem{SpeeWeber}
W.~J.~P.~Spee and J.~H.~Weber,
``Bounds on the maximum cardinality of indel and substitution correcting codes,''
\emph{IEEE Transactions on Molecular, Biological, and Multi-Scale Communications},
vol.~10, no.~2, pp.~349--358, 2024. doi:10.1109/TMBMC.2024.3388971.

\bibitem{Tenengolts1984}
G.~Tenengolts,
``Nonbinary codes, correcting single deletion or insertion,''
\emph{IEEE Transactions on Information Theory}, vol.~30, no.~5,
pp.~766--769, 1984.

\bibitem{Ukkonen}
E.~Ukkonen,
``Approximate string-matching with $q$-grams and maximal matches,''
\emph{Theoretical Computer Science}, vol.~92, no.~1, pp.~191--211, 1992.

\bibitem{Varshamov1957}
R.~R.~Varshamov,
``Estimate of the number of signals in error correcting codes,''
\emph{Dokl. Akad. Nauk SSSR}, vol.~117, no.~5,
pp.~739--741, 1957 (in Russian).
\url{https://www.mathnet.ru/eng/dan22571}.

\bibitem{VT}
R.~R.~Varshamov and G.~M.~Tenengolts,
``Codes which correct single asymmetric errors,''
\emph{Avtomatika i Telemekhanika}, vol.~26, no.~2, pp.~288--292, 1965;
English translation in \emph{Automation and Remote Control}, vol.~26,
no.~2, pp.~286--290, 1965.

\bibitem{WangKongYaakobiDuman}
C.~Wang, X.~Kong, E.~Yaakobi, and T.~M.~Duman,
``Upper bounds on multiple $b$-burst deletion-correcting codes,''
in \emph{Proc.\ IEEE International Symposium on Information Theory (ISIT)},
Guangzhou, China, 2026, pp.~1--6. doi:10.1109/ISIT62367.2026.11653749.

\bibitem{YeEtAlBursts}
Z.~Ye, Y.~Sun, W.~Yu, G.~Ge, and O.~Elishco,
``Codes correcting two bursts of exactly $b$ deletions,''
\emph{IEEE Transactions on Information Theory}, vol.~72, no.~4,
pp.~2083--2098, 2026. doi:10.1109/TIT.2025.3607876.

\bibitem{YekhaninDumer}
S.~Yekhanin and I.~Dumer,
``Long nonbinary codes exceeding the Gilbert--Varshamov bound for any fixed distance,''
\emph{IEEE Transactions on Information Theory}, vol.~50, no.~10,
pp.~2357--2362, 2004. doi:10.1109/TIT.2004.834744.


\end{thebibliography}
\end{document}